\documentclass[journal]{IEEEtran}
\usepackage{cite}

\ifCLASSINFOpdf
   \usepackage[pdftex]{graphicx}
   \graphicspath{{img/pdf/}{img/jpeg/}}
   \DeclareGraphicsExtensions{.pdf,.jpeg,.png}
\else
   \usepackage[dvips]{graphicx}
   \graphicspath{{img/eps/}}
   \DeclareGraphicsExtensions{.eps}
\fi

\usepackage[cmex10]{amsmath}
\usepackage{booktabs}
\usepackage{amsfonts}
\usepackage{amssymb}
\usepackage{mathrsfs}
\usepackage{bigints}
\usepackage{mathtools}
\usepackage[keeplastbox,nospread]{flushend}
\usepackage[tight,footnotesize]{subfigure}
\usepackage{wasysym}
\usepackage{color}
\usepackage{epsfig} 
\usepackage{epstopdf}
\usepackage{float}
\usepackage{amsthm}
\usepackage{mathtools}
\usepackage{relsize}
\definecolor{light-gray}{gray}{0.8}
\def\nb0{{\mathbf{0}}}
\def\nb1{{\mathbf{1}}}

\def\ncalF{{\mathcal{F}}}

\def\ncalP{{\mathcal{P}}}

\newtheorem{theorem}{Theorem}

\newtheorem{cor}{Corollary}

\def\E{\mathbb{E}}

\def\P{\mathbb{P}}

\def\R{\mathbb{R}}

\def\sir{\mathtt{SIR}}

\def\b {b}

\newcommand{\ud}{\, \mathrm{d}}
 
\newtheorem{corollary}[cor]{Corollary}

\begin{document}
\title{Effect of Retransmissions on Optimal Caching in Cache-enabled Small Cell Networks}
\author{\IEEEauthorblockN{Shankar Krishnan, Mehrnaz Afshang, and Harpreet S. Dhillon}
\thanks{The authors are with Wireless@VT, Department of ECE, Virginia Tech, Blacksburg, VA, USA. Email: \{kshank93, mehrnaz, hdhillon\}@vt.edu.}}
\maketitle
\begin{abstract}
Caching {\em popular content} in the storage of small cells is being considered as an efficient technique to complement limited backhaul of small cells in ultra-dense heterogeneous cellular networks. Limited storage capacity of the small cells renders it important to determine the optimal set of files (cache) to be placed on each small cell. In contrast to prior works on optimal caching, in this work we study the effect of {\em retransmissions} on the optimal cache placement policy for both static and mobile user scenarios. With the popularity of files modeled as a Zipf distribution and a maximum $n$ transmissions, i.e., $n-1$ retransmissions, allowed to receive each file, we determine the optimal caching probability of the files that maximizes the \emph{hit probability}. Our closed-form optimal solutions concretely demonstrate that the optimal caching probabilities are very sensitive to the number of retransmissions. 

\end{abstract}
\IEEEpeerreviewmaketitle
\begin{IEEEkeywords}
Stochastic geometry, small cell caching, retransmissions, mobility, hit probability, optimal cache.
\end{IEEEkeywords}
\allowdisplaybreaks
\section{Introduction}
Aggressive reuse of spectrum through dense deployment of {\em small cell base stations (SCBSs)} and caching popular content in their storage can help address growing capacity demands and reduce backhaul loads \cite{SCCaching}. The basic idea behind small cell caching is to download popular content (mainly video) automatically in the cache of SCBSs at off-peak hours, which can then be delivered to the users during peak hours. However, the key challenge in designing such a cache-enabled network is to determine the content that should be cached at each SCBS. Among the entire content available in the internet, only a small fraction of the total content, termed {\em popular content}, is accessed by a large fraction of users \cite{cha2007tube}, thus allowing us to focus primarily on the {\em library} of popular files to cache in SCBSs. However due to limited capacities of cache storage units, one can only cache a subset of the library on each SCBS. This necessitates the need to look into optimal caching of the popular content for maximum utilization of these caches and improve the overall network performance.

\subsubsection*{Prior work} Optimal caching of popular content in the storage of cache-enabled wireless networks has been studied for quite some time and can be broadly classified into two categories. The first line of work considers a scenario where each user has access (is in coverage) to at most one cache and the optimal cache placement results in storing the most popular files in the storage of all caches. To give an example, the authors in \cite{Rao} studied optimal cache placement in a Device-to-Device (D2D) assisted wireless caching network and showed that the most popular contents need to be cached more often in the network to maximize the offloading probabilities in this two-tier wireless caching system. The second line of work considers a multi-coverage scenario (user is covered by multiple caches) and studies optimal cache placement while exploiting the multiple covered caches as a larger distributed cache. Using tools from stochastic geometry, \cite{Bartec} shows that it is not always optimal to cache the most popular content everywhere in a multi-coverage scenario and proposes a probabilistic placement policy to maximize the user's \emph{hit probability}. Recently, the authors in \cite{OptCachingHetNet}  also studied optimal caching in heterogeneous cellular networks with different cache capabilities in each tier of the network and showed that an optimal content placement is significantly better than storing the most popular content everywhere. 

Despite increasing interest in the analysis of cache-enabled networks, all the prior works consider a single transmission scenario, which is not quite realistic because of the provision of retransmissions in actual systems. For wireless ad hoc networks, it is shown in \cite{nardelli2012optimal} and \cite{kaynia2011performance} that an increase in the number of allowed retransmissions decreases the link outage probability and thereby improves reliability. In this work, we extend the above argument to cache-enabled small cell networks by retransmitting the packets received in error to increase the chances of a cache hit. In particular, we determine the optimal caching strategies that maximize \emph{hit probability}. It should be noted that retransmissions may impact other metrics, e.g., it may increase energy consumption, increase latency, and lower spectral efficiency. However, the investigation of the effect of retransmissions on these metrics is not in the scope of this paper. More details are provided next.

\subsubsection*{Contributions} In this work, we determine the optimal caching policies that maximize hit probability for a typical user (static or mobile), which tries to receive its file of interest from the cache of SCBSs within a predefined number of maximum transmissions $n$. Here hit probability is the probability that the user is able to successfully receive the file within $n$ transmissions. As expected, it is seen that the hit probability increases with the number of retransmissions. In contrast to prior works on optimal caching, we also determine the optimal caching policies in static and mobile scenarios, and show that the optimal solutions in the two cases are significantly different. While the static case tends to cache the most popular files from the library on each SCBS, mobility de-emphasizes the importance of popularity of files and allows the SCBSs to cache content in a more {\em balanced} way. 

\section{System Model} \label{sec2}
We consider an ultra-dense cache-enabled SCBS network as illustrated in Fig. \ref{Fig: system model} and model the locations of the SCBSs by a homogeneous Poisson Point Processes (PPP) $\Phi$ with density $\lambda$\cite{dhillon2012modeling}. We assume each SCBS can cache a maximum of $L$ files, and the total number of files in the library is denoted by $K$. We denote by ${\tt P}_{{\rm R}_i}$ the probability that the $i^{th}$ file, ${\cal F}_i$, will be requested. We order the files based on their popularity, i.e., $i=1$ and $i=K$ correspond to the most popular and least popular files, respectively. The popularity of the files in the library is assumed to follow Zipf's law~\cite{cha2007tube}, i.e. ${\tt P}_{{\rm R}_i} = {i^{-\gamma}}/{\sum_{j=1}^K j^{-\gamma}}$, where $\gamma > 0$ is the Zipf parameter which governs the skewness of the popularity distribution.  We also assume that each SCBS caches file ${\cal F}_i$ with probability $b_i$ independently of the other SCBSs. Therefore, $\sum_{i=1}^K b_i = L$.
\subsubsection*{Retransmissions}
We consider a typical user that attempts to download its file of interest from the cache of SCBSs in an ultra dense network as shown in Fig. \ref{Fig: system model}. It is quite likely that a single attempt to download the file is unsuccessful, either because none of the SCBSs in the user's vicinity has cached the file (each SCBS caches only a fraction of the library) or due to poor channel conditions. In this work, we hence study the effect of retransmissions on the network performance for both static (same user location during all retransmission attempts) and mobile user scenarios. For the analysis of mobile users, we consider the popular {\em infinite mobility model} \cite{haenggi2013local, backhaul, UserMobility} in which a user is assumed to experience an independent realization of the point process for each transmission. This model is quite relevant for ultra dense networks where even a small displacement of a user may take it to a completely new {\em local neighborhood} of SCBSs.

For the setup studied in this paper, we assume that a typical user (can be mobile or static) tries to receive its file of interest ${\cal F}_i$ from the cache of SCBSs for a maximum of $n$ transmissions (or $n-1$ retranmissions). The signal to interference ratio ($\sir$) received at the typical user in the $k^{\rm th}$ transmission can be expressed as $\sir_{i,k} =  \frac{h_{xk} \|x\|^{-\alpha}}{\sum\limits_{y \in \Phi \backslash \{x\}} h_{yk} \|y\|^{-\alpha} },$
where $\{h_{xk} , h_{yk}\} \sim \exp(1)$ denote the Rayleigh fading channel gains from the serving device $x \in \Phi$ and interferer $\{y\}$ in the $k^{\rm th}$ transmission, and $\|\cdot \|^{-\alpha}$ is standard power-law pathloss with exponent $\alpha>2$. Here, we assume that the fading gains are independent across transmission attempts.
\subsubsection*{Cell Selection Policy}
To choose a serving SCBS for the typical user,  a straightforward choice would be to connect to the SCBS that maximizes its average received power at the receiver, agnostic of the cached file in that SCBS. This corresponds to the closest SCBS to the typical user. We term this policy as policy 1 ({\em cache-agnostic policy}) and denote it by ${\cal P} 1$. While this  policy is meaningful in general, it suffers from a disadvantage in cache-enabled networks: the closest SCBS is not guaranteed to have the file requested by the user. To address this, we also study policy 2 ({\em cache-aware policy}), denoted by ${\cal P} 2$, where the user connects to the closest SCBS that has the file it needs. This requires the knowledge of the cache of nearby SCBSs, which can be obtained by a centralized mechanism where macrocells can assist the users to connect to the SCBS which contains its file of interest. However since this SCBS may not always be in the vicinity of the user, the file transfer may not necessarily succeed due to poor channel conditions. Therefore, there are tradeoffs involved in choosing the cell selection policies.
\begin{figure}[!t]
\centering
  \includegraphics[width=1.0 \linewidth]{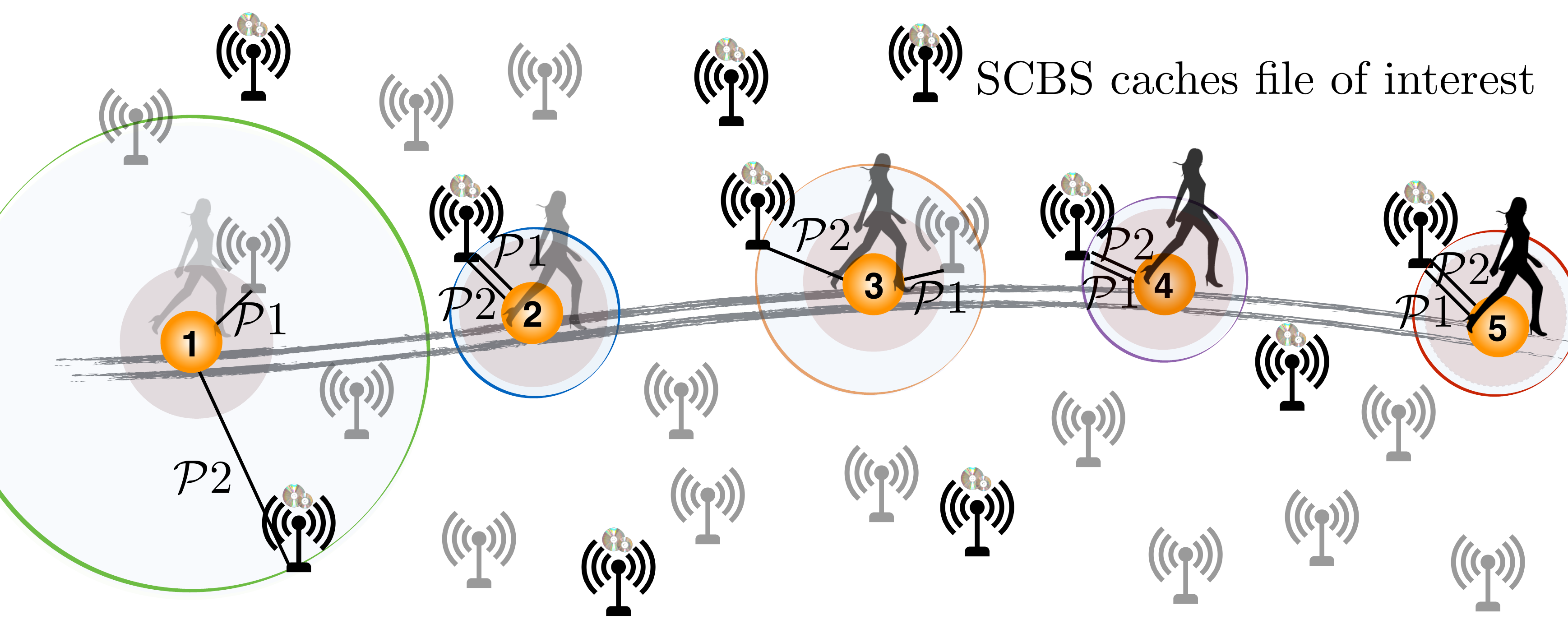}
  \caption{User tries to obtain the file of interest from the network using Policies 1 and 2 (${\cal P}1$ and ${\cal P}2$) while moving from location 1 through 5. Under ${\cal P}1$, user connects to the SCBS providing highest average received power (closest). Under ${\cal P}2$, user connects to the closest SCBS that has its file of interest. }
  \label{Fig: system model}
\end{figure}

For this system model, we introduce the optimal caching problem that maximizes the \emph{hit probability} under the two cell selection policies ${\cal P} 1$ and ${\cal P} 2$ in the next section.  

\section{Hit Probability} \label{Sec:sec3}

In a given transmission, a file is successfully received only when the user is in \emph{coverage} of a SCBS which has the file of interest in its cache. For a given modulation and coding threshold $T$, the coverage probability of file ${\cal F}_i$ in the  $k^{\rm th}$ transmission is given by $\P(\sir_{i,k} > T)$. Let $S_i$ be the event that file ${\cal F}_i$ is successfully received within $n$ transmissions with the \emph{success probability} denoted by $p_{{\tt s}_i}^{1|n} = \P(S_i)$. We assume i.i.d. fading over the $n$ transmissions and hence the success probability of file ${\cal F}_i$ in each transmission is the same and denoted by $p_{{\tt s}_i}$. 
 If all the $n$ transmissions in receiving the file are unsuccessful, the user is said to be in outage from file ${\cal F}_i$. Let the outage probability of file ${\cal F}_i$ in $n$ transmissions be given by $p_{{\tt o}_i}^{1|n} = 1 - p_{{\tt s}_i}^{1|n}$. The performance metric of interest is the {\em hit probability (HP)}, which we define mathematically as the sum of the probabilities of successfully receiving each file in the library within $n$ transmissions, weighted by their corresponding request probabilities and is expressed below.
\begin{align} \label{Hitprob}
{\tt P}_{\tt hit} = \sum_{i = 1}^K  {\tt P}_{{\rm R}_i} p_{{\tt s}_i}^{1|n} = \sum_{i = 1}^K  {\tt P}_{{\rm R}_i}(1-p_{{\tt o}_i}^{1|n}).
\end{align}
Our goal is to maximize the hit probability which will be defined for two scenarios - i) mobile user and ii) static user in the next two subsections. The optimal caching probabilities $\{b_i\}$ that maximize the HP for the two cell selection policies ${\cal P} 1$ and ${\cal P} 2$ will also be discussed for both scenarios. 
\subsection{Mobile user}
Under infinite mobility assumption, the probability of file success (or outage) in each transmission is independent of the previous transmissions. Hence, the outage (or success) probability of file ${\cal F}_i$ in $n$ transmissions is simply given as the product of outage (or success) probabilities in each transmission i.e. $p_{{\tt o}_i}^{1|n} = (1- {p}_{{\rm s}_i})^n$. The optimization problem to maximize the hit probability can be hence formulated as:
\begin{align} \label{Eq:OptFunctionHMU}
&\max_{\{b_i\}} \: \: \sum_{i = 1}^K  {\tt P}_{{\rm R}_i}(1-(1-{p}_{{\rm s}_i})^{n}), \\\notag & {\rm s.t.} \: \: \sum_{i = 1}^K b_i = L \: \text{and} \:  \: 0 \leq b_i \leq 1, \: i = 1, \dots K.
\end{align}
\subsubsection{Policy 1}

As described in section \ref{sec2}, under policy 1 user connects to the closest SCBS as that maximizes its average received signal strength. A successful reception of the file hence depends on the probability that it is available in the cache of the closest SCBS and that the user is in {\em coverage} of the closest SCBS.  Mathematically, the probability of successfully receiving file ${\cal F}_i$ in a given transmission is thus the product of its caching probability $b_i$ and coverage probability under policy 1 denoted by ${p}_{{\rm c}_{i}}^{(1)}$ i.e. ${p}_{{\rm s}_i} = b_i {p}_{{\rm c}_{i}}^{(1)}$.
It is worth noting that this coverage probability ${p}_{{\rm c}_{i}}^{(1)}$(when user connects to the closest SCBS) is  independent of the density of SCBSs under interference limited regime and has been derived in~\cite{AndBacJ2011} for a similar cellular downlink problem. The probability of successfully receiving file ${\cal F}_i$ in a given transmission under policy 1 is therefore,
\begin{align} 
{p}_{{\rm s}_i} = b_i {p}_{{\rm c}_{i}}^{(1)} &= \frac{b_i}{1 + \rho_1(T,\alpha)}, \label{Eq:SuccProbP1} \\
\text{where} \: \rho_1(T,\alpha) &= T^{2/\alpha} \int_{T^{-2/\alpha}}^{\infty} \frac{\ud u }{1+u^{\alpha /2}}. \label{Eq:rho1}
\end{align}
It can be shown that the optimization problem (\ref{Eq:OptFunctionHMU}) is concave and hence the Karush-Kuhn-Tucker (KKT) conditions provide necessary and sufficient conditions for optimality. The Lagrangian function corresponding to Problem (\ref{Eq:OptFunctionHMU}) is 
\begin{align}
& \notag {\cal L}(\mathbf{b},\nu,\boldsymbol{\mu},\mathbf{w}) = \sum_{i = 1}^K {\tt P}_{{\rm R}_i}(1-(1-b_i{p}_{{\rm c}_{i}}^{(1)})^{n}) + \nu (\sum_{i = 1}^K b_i - L) \\\notag & \qquad \qquad - \sum_{i=1}^K \mu_i b_i  + \sum_{i=1}^{K} w_i(b_i -1),\\\notag
&\text{where} \: \boldsymbol{\mu},\mathbf{w} \in \R_+^K \: \text{and} \: \nu \in \R. 
\end{align}
Let $\mathbf{b}^*, \nu^*, \boldsymbol{\mu}^*$ and $\mathbf{w}^*$ be primal and dual optimal. The KKT conditions for Problem (\ref{Eq:OptFunctionHMU}) state that
\begin{align}
\sum_{i = 1}^K b_i^* &= L, \label{constraint summation}  \\ 
0 \leq \b_i^* \leq 1, \mu_i^* \geq 0, \ w_i^* \geq 0, \mu_i^*b_i^* &= 0, \forall i = 1, \dots K  \label{constraint no 1} \\ w_i^*(b_i^*-1) &= 0, \forall i = 1, \dots K,   \label{constraint no 2} \\\notag
{\tt P}_{{\rm R}_i}n(1-b_i^*{p}_{{\rm c}_{i}}^{(1)})^{n-1}{p}_{{\rm c}_{i}}^{(1)} \\+ \nu^* - \mu_i^* +w_i^* &= 0, \forall i = 1, \dots K. \label{constraint no 3}
\end{align}
The optimal cache placement under policy 1 is given next. 
\begin{theorem} \label{Theorem 1}
The optimal caching probability of file ${\cal F}_i$ denoted by $b_i^{*}$, that maximizes the HP (with a maximum $n$ transmissions) for a mobile user under policy 1, is given by,
\begin{align}
b_i^{*} &= \left\{
     \begin{array}{lr}
       0, \qquad \qquad \nu^* < -{\tt P}_{{\rm R}_i}n{p}_{{\rm c}_{i}}^{(1)}\\
       1, \qquad  \qquad \nu^* > -{\tt P}_{{\rm R}_i}n{p}_{{\rm c}_{i}}^{(1)}(1-{p}_{{\rm c}_{i}}^{(1)})^{n-1}\\
      \frac{1}{{p}_{{\rm c}_{i}}^{(1)}}\big[ 1 - \big(\frac{-v^*}{{\tt P}_{{\rm R}_i}n{p}_{{\rm c}_{i}}^{(1)}}\big)^{\frac{1}{n-1}}\big] , \qquad   \text{otherwise}
      \end{array}
   \right.
\end{align}
where $\nu^* = -{\tt P}_{{\rm R}_i}n(1-b_i^*{p}_{{\rm c}_{i}}^{(1)})^{n-1}{p}_{{\rm c}_{i}}^{(1)}$ can be obtained as the solution of the constraint $\sum_{i = 1}^K b_i^* = L$. 
\end{theorem}
\begin{proof}
From (\ref{constraint no 1}) and  (\ref{constraint no 3}), we have
\begin{align}
w_i^* = b_i^*[-{\tt P}_{{\rm R}_i}n(1-b_i^*{p}_{{\rm c}_{i}}^{(1)})^{n-1}{p}_{{\rm c}_{i}}^{(1)} - \nu^*],
\end{align}
which when inserted into (\ref{constraint no 2}) gives
\begin{align} \label{constraint no 4}
b_i^*(b_i^*-1)[-{\tt P}_{{\rm R}_i}n(1-b_i^*{p}_{{\rm c}_{i}}^{(1)})^{n-1}{p}_{{\rm c}_{i}}^{(1)} - \nu^*] = 0.
\end{align}
From (\ref{constraint no 4}), we see that $0 < b_i^{*}<1$ only if,
\begin{align} 
\nu^* = -{\tt P}_{{\rm R}_i}n(1-b_i^*{p}_{{\rm c}_{i}}^{(1)})^{n-1}{p}_{{\rm c}_{i}}^{(1)}.
\end{align}
Since we know that $0 \leq \b_i^* \leq 1$, this implies that
\begin{align} 
\nu^* \in [-{\tt P}_{{\rm R}_i}n{p}_{{\rm c}_{i}}^{(1)},-{\tt P}_{{\rm R}_i}n{p}_{{\rm c}_{i}}^{(1)}(1-{p}_{{\rm c}_{i}}^{(1)})^{n-1}].
\end{align}
For the above interval, solving for $\nu^*$ using the constraint $\sum_{i = 1}^K b_i^* = L$, we get,
\begin{align}
\notag \sum_{i = 1}^K  \frac{1}{{p}_{{\rm c}_{i}}^{(1)}}\big[ 1 - \big(\frac{-v^*}{{\tt P}_{{\rm R}_i}n{p}_{{\rm c}_{i}}^{(1)}}\big)^{\frac{1}{n-1}}\big] &= L \\ 
\bigg(\frac{-v^*}{n{{p}_{{\rm c}}^{(1)}}}\bigg)^{\frac{1}{n-1}} &\stackrel{(a)} = \frac{K-L{p}_{{\rm c}}^{(1)}}{ \sum_{j = 1}^K \big(\frac{1}{{\tt P}_{{\rm R}_j}}\big)^{\frac{1}{n-1}}} \label{Eq:Golden1}
\end{align}
where (a) results by using ${p}_{{\rm c}_{i}}^{(1)} = {p}_{{\rm c}}^{(1)},  \forall i= 1, \dots K$ and rearranging few terms. Also, it can be seen that for $\nu^* < -{\tt P}_{{\rm R}_i}n{p}_{{\rm c}_{i}}^{(1)}, b_i^* = 0$ and if $v^* > -{\tt P}_{{\rm R}_i}n{p}_{{\rm c}_{i}}^{(1)}(1-{p}_{{\rm c}_{i}}^{(1)})^{n-1}$, $b_i^* = 1$.
\end{proof}
In order to provide intuition, we specialize the above result to the simple case of 2 files in the library ($K = 2$) and unitary storage space ($L = 1$) in the SCBS.
\begin{corollary} \label{Cor1}
The optimal value ($b_1^*,b_2^*$) obtained by solving the optimization problem (\ref{Eq:OptFunctionHMU}) for $K = 2$ is
\begin{align}
b_1^{*} &= \left\{
     \begin{array}{lr}
       1, \qquad \qquad n < 1+ \frac{\gamma}{\log_2\big(\frac{1}{{1-{p}_{{\rm c}_{i}}^{(1)}}}\big)}\\
       \frac{a-1+{p}_{{\rm c}_{i}}^{(1)}}{(a+1){p}_{{\rm c}_{i}}^{(1)}} , \qquad   \text{otherwise}
      \end{array}
   \right. ,
\end{align}
where $a = 2^{\frac{\gamma}{n-1}}$, $\gamma$ is the Zipf parameter and $b_2^* = 1 - b_1^*$.
\end{corollary}
\begin{proof}
See Appendix A.
\end{proof}

From Corollary \ref{Cor1}, we can see that it is optimal to cache only the most popular file ${\cal F}_1$ till a certain number of transmissions, however an optimal cache (cache both ${\cal F}_1$ and ${\cal F}_2$) exists as the number of transmissions increase. Also, for $n =1 $ (single transmission), it is always optimal to cache the most popular file in a 2-file library scenario.

For the simplicity of exposition, we will henceforth consider $L = 1$ while deriving the optimal solutions. A similar approach can be undertaken to derive the optimal solutions for the generic $L$ storage system.

\subsubsection{Policy 2}
In this policy, the user connects to the closest SCBS which has the file of interest in its cache. The key difference in the mathematical formulation under this policy compared to ${\cal P} 1$ is that the success probability is not weighted by caching probability of the file as the user is always connected to the SCBS that has the file of interest in its cache. Hence, the success probability of obtaining file ${\cal F}_i$ under policy 2 is the same as its coverage probability denoted by ${p}_{{\rm c}_{i}}^{(2)}$, which has been derived for a similar scenario in \cite[Theorem 1]{KriDhi} as: 
\begin{align} 
&{p}_{{\rm s}_i} = {p}_{{\rm c}_{i}}^{(2)} = \frac{b_i}{b_i + \rho_1(T,\alpha)+(1-b_i)\rho_2(T,\alpha)},   \label{SuccProbP2} 
\\ &\text{where} \: \rho_2(T,\alpha) = T^{2/\alpha} \int_{0}^{T^{-2/\alpha}} \frac{\ud u}{1+u^{\alpha/2}},    \label{rho 2}
\end{align}
and $\rho_1(T,\alpha)$ is defined in (\ref{Eq:rho1}).
\normalsize
\\ The solution of the optimization problem (\ref{Eq:OptFunctionHMU}) under policy 2 can be determined on the same lines as Theorem \ref{Theorem 1} (policy 1) by using the success probability ${p}_{{\rm s}_i}$ defined above. The optimal solution is stated below in Theorem \ref{Theorem 2}.
\begin{theorem} \label{Theorem 2}
The optimal caching probability of file ${\cal F}_i$ denoted by $b_i^{*}$, that maximizes the hit probability (with a maximum $n$ transmissions) for a  mobile user under policy 2, is given by,
\begin{align}
b_i^{*} &= \left\{
     \begin{array}{lr}
       0, \qquad \qquad \nu^* < \frac{-{\tt P}_{{\rm R}_i}n}{C}\\
       1, \qquad  \qquad \nu^* > \frac{-{\tt P}_{{\rm R}_i}nC(B+C-1)^{n-1}}{(B+C)^{n+1}}\\
      \phi(\nu^*) , \qquad   \text{otherwise}
      \end{array}
   \right. ,
\end{align}
where $\phi(v^*)$ is the solution over $b_i$ of
\begin{align} \label{Eq:Polynomial Equality P2}
\frac{{\tt P}_{{\rm R}_i}nC((B-1)b_i+C)^{n-1}}{(Bb_i+C)^{n+1}}+\nu^* = 0,
\end{align} 
$B = 1-\rho_2(T,\alpha)$, $C = \rho_1(T,\alpha) + \rho_2(T,\alpha)$ and $\nu^*$ can be obtained as the solution of the constraint $\sum_{i = 1}^K b_i^* = 1$.
\end{theorem}
The success probability ${p}_{{\rm s}_i}$ for policy 2 (given by (\ref{SuccProbP2})), being more complicated than policy 1 (given by (\ref{Eq:SuccProbP1})) makes the
optimal solution in this case harder to obtain. This requires solving the polynomial equalities of the form (\ref{Eq:Polynomial Equality P2}), which may not give simple closed form solutions. Therefore, we limit our further discussion in providing the optimal solution only for the extreme cases ($n = 1$ and $n \rightarrow \infty)$ for a user following $\ncalP 2$.
\begin{corollary} (Single transmission - Policy 2) \label{Cor2}
The optimal caching probability of file ${\cal F}_i$ denoted by $b_i^{*}$, that maximizes the HP for a mobile user under $\ncalP 2$ with a maximum of one attempt ($n =1$), is given by,
\begin{align}
&b_i^{*} = \bigg[\frac{\sqrt{\frac{{\tt P}_{{\rm R}_i}}{\epsilon}}- (\rho_1(T,\alpha)+\rho_2(T,\alpha))}{1-\rho_2(T,\alpha)}\bigg]^{+}, i = 1, \dots K, 
\end{align}
where $ [x]^+ = {\tt max}(0,x), \: \sqrt{\epsilon} = \frac{\sum_{i=1}^{K^*}\sqrt{{\tt P}_{{\rm R}_i}}}{(K^*-1)\rho_1(T,\alpha)+K^*\rho_2(T,\alpha)+1}$ and $K^*, \: 1 \leq K^* \leq K$, satisfies the constraint that $0 \leq b_i^* \leq 1$.  Here $\rho_1(T,\alpha)$ and  $\rho_2(T,\alpha)$ are defined in (\ref{Eq:rho1}) and (\ref{rho 2}) respectively. 
\end{corollary}
\begin{proof}
The result follows by substituting $n= 1$ (single attempt) in Theorem \ref{Theorem 2}, solving for $v^*$ using the constraint $\sum_{i = 1}^K b_i^* = 1$ and simple mathematical manipulation. The detailed proof is skipped due to space constraints.
\end{proof}
\begin{corollary} (Large number of transmissions) \label{Cor3}
For the mobile user scenario (policy 1 and 2) with the number of transmissions approaching infinity (asymptotically), it is optimal to cache the files in the library evenly  i.e. $\lim\limits_{n \to \infty} {b_i}^* = \frac{1}{K} $, where $K$ is the total number of popular files in the library.
\end{corollary}
\begin{proof}
As evident from (\ref{Eq:SuccProbP1}) and (\ref{SuccProbP2}), $p_{{\tt s}_i}$ is a monotonically increasing function of $b_i$ for both caching policies, say $p_{{\tt s}_i} = f(b_i)$. The optimization function (\ref{Eq:OptFunctionHMU}) can be hence generalized as: 
\begin{align} \label{Eq:OptFunHMUNew}
\max_{\{b_i\}} \: \: \sum_{i = 1}^K  {\tt P}_{{\rm R}_i}(1-(1-f(b_i))^{n}), \:  {\rm s.t.} \: \: \sum\limits_{i = 1}^K b_i = L.
\end{align}
Taking the derivative of (\ref{Eq:OptFunHMUNew}) w.r.t. $\{b_i\}_{i=1...K}$, we obtain
\begin{align*}
{\tt P}_{{\rm R}_i}n(1-(1-f(b_i))^{n-1}) &= 0 \qquad \forall \: i = 1,...K, \: \: \text{i.e.}, \\\notag
{\tt P}_{{\rm R}_i}n(1-(1-f(b_i))^{n-1}) & = {\tt P}_{{\rm R}_j}n(1-(1-f(b_j))^{n-1}), \forall i \neq j \\\notag
\frac{1-f(b_i^*)}{1-f(b_j^*)} &= \left(-{\frac{f'(b_j^*)P_{R_j}}{f'(b_i^*)P_{R_i}}}\right)^{\frac{1}{n-1}}, \qquad \forall i \neq j
\end{align*}
For $n \rightarrow \infty$, we hence obtain $f(b_i^*) = f(b_K^*) \: $, or equivalently $ b_i^* = b_j^*$. With $\sum_{i = 1}^K b_i^* = L$, it is therefore optimal to cache evenly i.e. $b_i^{*} = L/K$ for a mobile user with a large number of retransmissions.
\end{proof}

\subsection{Static user}
In this scenario, the user is static at a given location in the network and tries to receive the file of interest within $n$ transmissions. As the user is static, it remains in the same local neighbourhood of transmitters across the $n$ transmissions resulting in temporal correlation of the interferers. As a result, the probability of successfully receiving the file in a given transmission depends on the success probability in the previous transmissions. Let $S_{i,k}$ be an event that denotes that file ${\cal F}_i$ is in coverage during the $k^{\rm th}$ transmission attempt. The probability that file ${\cal F}_i$ is in coverage (atleast) once in $n$ transmissions,  denoted by $p_{{\tt c}_i}^{1|n}$  is given as
\begin{align} 
\notag p_{{\tt c}_i}^{1|n} &= \P(\cup_{k=1}^{n} S_{i,k}) = \P( \cup_{k=1}^{n} \sir_{i,k} > T) \\ &\stackrel{(a)}= \sum_{k=1}^n   { {n} \choose {k}} (-1)^{k+1} {\tt P}_{i,k} \label{JointSuccessStatic},
\end{align}
where (a) follows from the inclusion-exclusion principle and ${\tt P}_{i,k} = \P( \cap_{j=1}^{k} \sir_{i,j} > T)$ is defined as the joint coverage probability of file ${\cal F}_i$ in $k$ transmissions.

As was the case in the mobile user scenario, we define an optimization problem for maximizing the hit probability for the two caching policies ${\cal P} 1$ and ${\cal P} 2$ for a static user.
\subsubsection{Policy 1}
Taking into account that the file of interest ${\cal F}_i$ is cached only with probability $b_i$ in the closest SCBS, we multiply coverage probability $p_{{\tt c}_i}^{1|n}$ by $b_i$ to obtain the success probability in $n$ transmissions i.e.  $p_{{\tt s}_i}^{1|n} =  b_i p_{{\tt c}_i}^{1|n}$. Using (\ref{JointSuccessStatic}) in the above result and substituting in (\ref{Hitprob}), we obtain the following optimization problem.
\begin{align} \label{OptFunStaticP1}
\max_{\{b_i\}} \: &\sum_{i=1}^{K}{\tt P}_{{\rm R}_i}\sum_{k=1}^n  { {n} \choose {k}} (-1)^{k+1} b_i {{\tt P}^{(1)}_{i,k}}, \: {\rm s.t}  &\sum_{i=1}^{K} {b_i}=1,
\end{align}
where ${{\tt P}^{(1)}_{i,k}}$, the joint coverage probability of file ${\cal F}_i$ in $k$ transmissions under $\ncalP 1$ is derived in Appendix B and given below, i.e. ${{\tt P}^{(1)}_{i,k}} =$
\small
\begin{align} \label{Eq:JointSuccProbP1}
\int\limits_0^{\infty }\exp \bigg(-2 \pi \lambda \int\limits_{r_1}^{\infty} \bigg(1-\big(\frac{u^\alpha}{T r^{\alpha}+u^{\alpha}}\big)^{k}\bigg)u {\rm d} u\bigg) f_{R_1}(r_1) {\rm d} r_1,
\end{align}
\normalsize
where $R_1$ denotes the distance of the closest SCBS from the typical user (or distance of the closest point of the PPP $\Phi$ of intensity $\lambda$). The distribution of $R_1$ is hence given from the null probability of a PPP as $f_{R_1}(r_1) = 2\lambda \pi r_1 e^{-\lambda \pi r_1^2}$ ~\cite{HaeB2013}.
\subsubsection{Policy 2}
In policy 2, as the user is always connected to the SCBS that has the file of interest in its cache, the success probability of obtaining ${\cal F}_i$ is the same as its coverage probability i.e $p_{{\tt s}_i}^{1|n} =  p_{{\tt c}_i}^{1|n}$.  Thereby, simply using (\ref{JointSuccessStatic}) in (\ref{Hitprob}), we obtain the following optimization problem for policy 2.
\begin{align} \label{OptFunSUP2}
\max_{\{b_i\}} \: &\sum_{i=1}^{K}{\tt P}_{{\rm R}_i}\sum_{k=1}^n  { {n} \choose {k}} (-1)^{k+1} {{\tt P}^{(2)}_{i,k}}, \: {\rm s.t} \: &\sum_{i=1}^{K} {b_i}=1,
\end{align}
where ${{\tt P}^{(2)}_{i,k}}$, the joint coverage probability of file ${\cal F}_i$ in $k$ transmissions under policy 2 is derived by proceeding similar to policy 1 and is stated below. 
\small
\begin{align*} \label{Eq:JointSuccProbP2}
&\notag {{\tt P}^{(2)}_{i,k}} = \int_0^{\infty }\exp \bigg(-2 \pi (1-b_i) \lambda \int_0^{\infty} \bigg(1-\big(\frac{u^\alpha}{T {r_2}^{\alpha}+u^{\alpha}}\big)^{k}\bigg)u {\rm d} u\bigg) \\ &\exp \bigg(-2 \pi b_i \lambda \int_{r_2}^{\infty} \bigg(1-\big(\frac{u^\alpha}{T {r_2}^{\alpha}+u^{\alpha}}\big)^{k}\bigg)u {\rm d} u\bigg) f_{R_2}(r_2) {\rm d} r_2,
\end{align*}
\normalsize
where $R_2$ denotes the distance of the closest SCBS that has the file of interest ${\cal F}_i$ in its cache. As file ${\cal F}_i$ is cached with probability $b_i$ in the network, the distribution of $R_2$ is thus given by the closest point of the PPP of intensity $\b_i \lambda$ and its distribution is given by $f_{R_2}(r_2) = 2b_i \lambda \pi r_2 e^{-b_i \lambda \pi r_2^2}$. The key difference in the analysis of policy 2 is that the interference field is now divided into two regions: i) Interference from those SCBSs with the file of interest ${\cal F}_i$ in their cache, constituting a PPP of intensity $b_i\lambda$ outside a radius $r_2$ (closest distance of ${\cal F}_i$) and ii) Interference from all other SCBSs (not having ${\cal F}_i$) constituting a PPP of intensity $(1-b_i)\lambda$ in $\R^2$. 

The optimal caching probabilities of a static user under policy 1 and 2 i.e. the solutions of the optimization problem (\ref{OptFunStaticP1}) and (\ref{OptFunSUP2}) can be obtained by proceeding similar to the mobile user scenario (Theorem \ref{Theorem 1} and  \ref{Theorem 2}). Further insights on the optimal caching probabilities for a static user under both caching policies are provided in section \ref{sec5}.


\begin{figure}[t!]
\centering
\begin{subfigure}{}
\centering
\includegraphics[width=0.72 \linewidth]{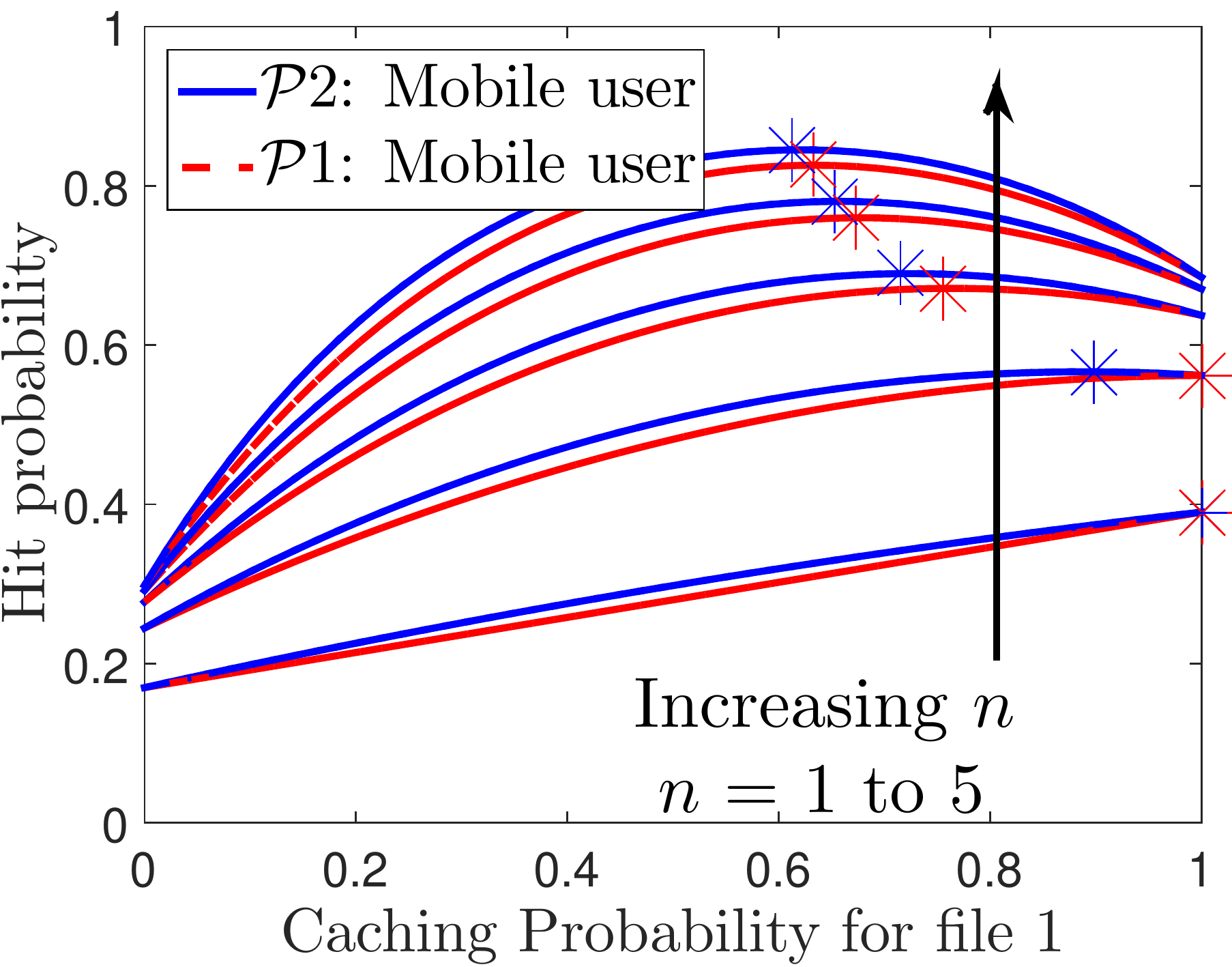}
\label{Fig:Fig: P1 vs P2}
\end{subfigure}
\begin{subfigure}{}
\centering
\includegraphics[width=0.72 \linewidth]{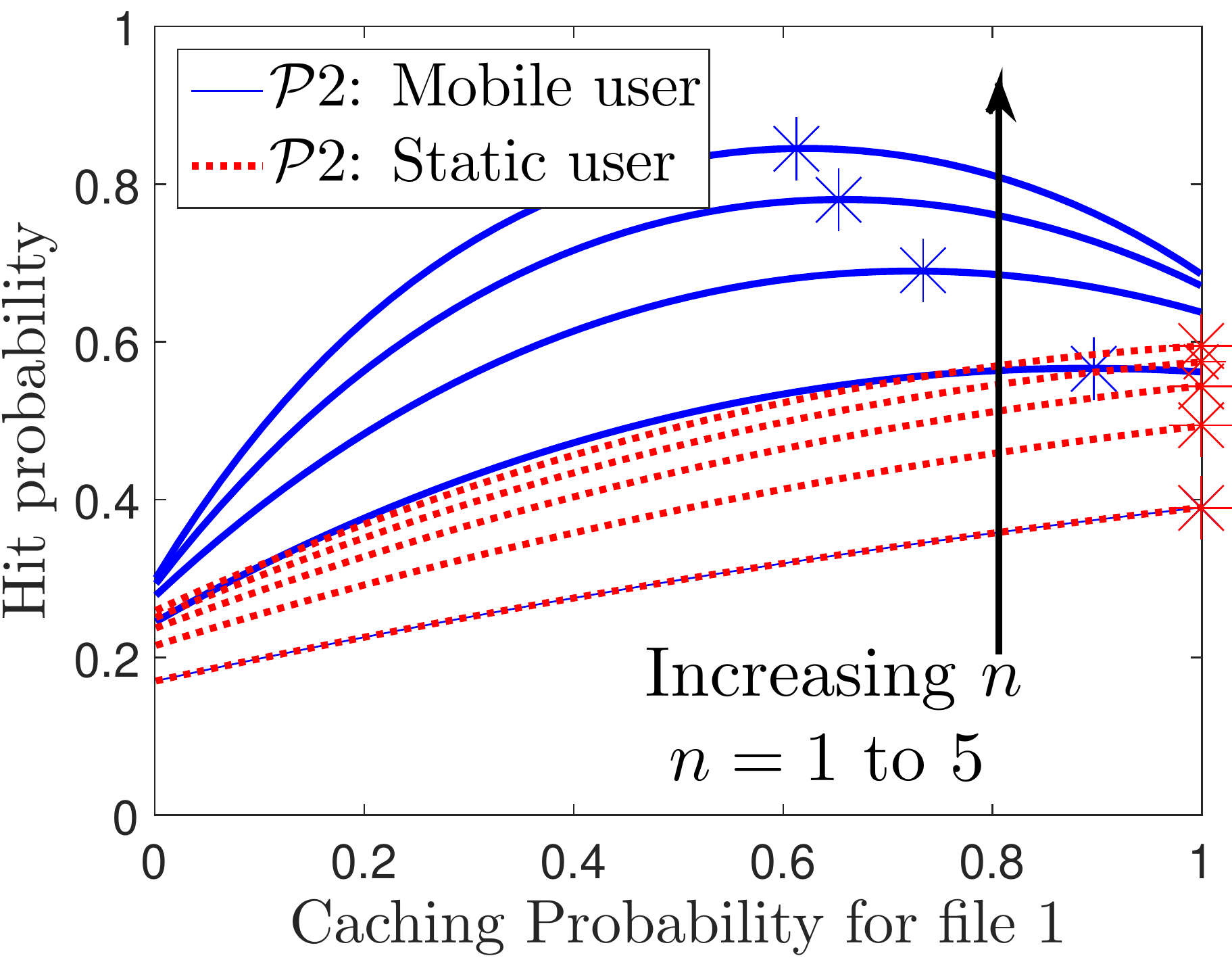}
\label{Fig:Hit probability P2}
\end{subfigure}
\caption{ Effect of cache gathering policy (policy 1 and 2) and mobility (static and mobile user) on the hit probability for varying number of transmissions. ($K = 2$, $L = 1$)}
\label{Fig:EffectMob}
\end{figure}
\section{Results and Discussion} \label{sec5}
For the purpose of numerical results, we consider the Zipf parameter $\gamma = 1.2$ and an SIR threshold $\beta = 0 \: {\tt dB}.$. It is to be noted that in Fig. \ref{Fig:EffectMob}, the asterisk denotes the optimal caching probability. 

\subsubsection{Comparison of policies $\ncalP 1$ and $\ncalP 2$}
Fig. \ref{Fig:EffectMob} (\emph{left.}) compares the two policies in terms of HP for different caching probabilities of file 1, $\ncalF_1$ and varying number of transmissions $n$. It can be seen that the optimal caching probabilities are shifted slightly towards the left in case of policy $\ncalP 2$ and has a slightly higher HP compared to policy $\ncalP 1$. This behaviour can be attributed to the policy mechanism itself, with $\ncalP 2$ exhibiting a higher HP due to the available information of the cache locations. Also, the shift in optimal caching probabilities assert that it is not necessary to cache the most popular file (file $\ncalF_1$) predominantly in the network and allows the SCBSs to cache the lesser requested files to a comparatively larger extent.  Finally, as expected, the HP for both policies $\ncalP 1$ and $\ncalP 2$ increase with the increasing number of transmissions $n$.
\subsubsection{Comparison of the mobile and static cases as a function of $n$ for policy $\ncalP 2$}
Fig. \ref{Fig:EffectMob} (\emph{right}.) compares the HP for the mobile and static cases when a user following policy $\ncalP 2$ has a maximum $n$ transmissions to obtain the file of interest in a 2-file library case ($K = 2$) with unitary storage ($L = 1$). While in the mobile case the retransmissions are made at different locations, in the static case the location of the user is not changed and hence all the retransmissions are made to the same user location. However, we assume that the fading gains are independent across time slots in which these retransmissions are made. The temporal diversity due to independent fading gains increases the HP with $n$, even for the static case. This trend is evident in Fig. \ref{Fig:EffectMob} (\emph{right}.), where the static case is shown by the dotted lines.  As expected, the HP for the mobile user is significantly higher compared to the static user and the gap widens further with the increasing $n$. Similar observations can be for policy ${\cal P}1$ as well.
\subsubsection{Effect of retransmissions on optimal cache placement strategy}
Considering a library of $K=3$ files, Fig. \ref{Fig:ThreeFile} depicts the optimal caching probabilities for a SCBS network with cache size $L$. For smaller number of retransmissions (small $n$), ``{\em Cache the most popular content}" strategy seems to be the optimal cache placement strategy. In other words, it is optimal to cache the $L$ most popular files out of the $K$ files in the library for smaller number of retransmissions. In case of large retransmissions, the optimal caching probabilities asymptotically approaches $L/K$, as we discussed in Corollary \ref{Cor3}.
\subsubsection{Effect of library size ($K$) on the hit probability}
Fig. \ref{Fig: EffectK} depicts the optimal hit probability in a network (for a mobile user following policy 1) with $K$ files in the library and cache size $L = 1$. As evident from the figure, the hit probability decreases with increase in number of files in the library. Larger the number of files in the library, fewer the chances of a file hit from the SCBS cache with a certain cache size. Also, one can observe that the optimal hit probability increases with the number of retransmissions as discussed before.
\begin{figure}[t!]
\centering
\begin{subfigure}{}
\centering
\includegraphics[width=0.72 \linewidth]{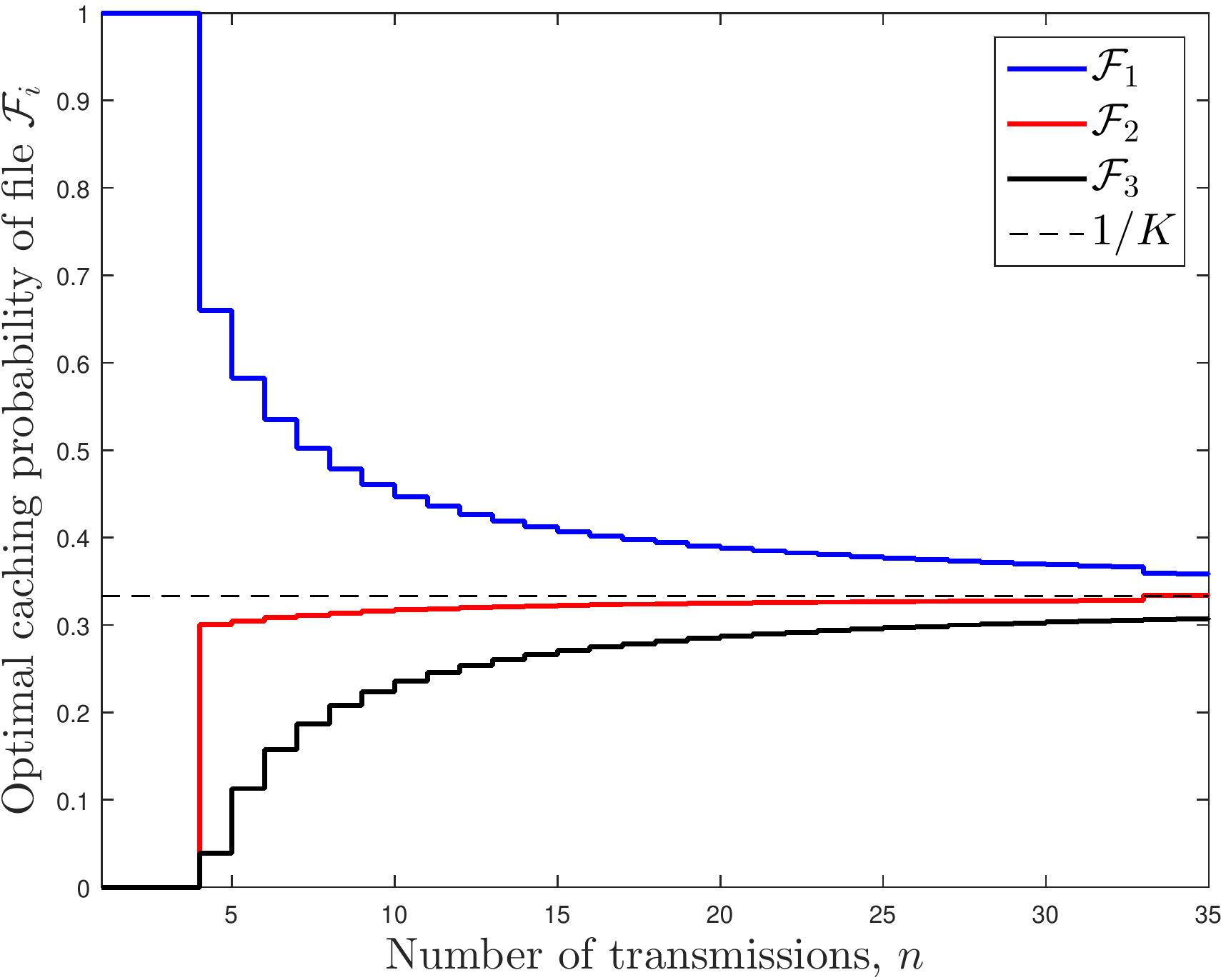}
\label{Fig:ThreeFileL1}
\end{subfigure}
\begin{subfigure}{}
\centering
\includegraphics[width=0.72 \linewidth]{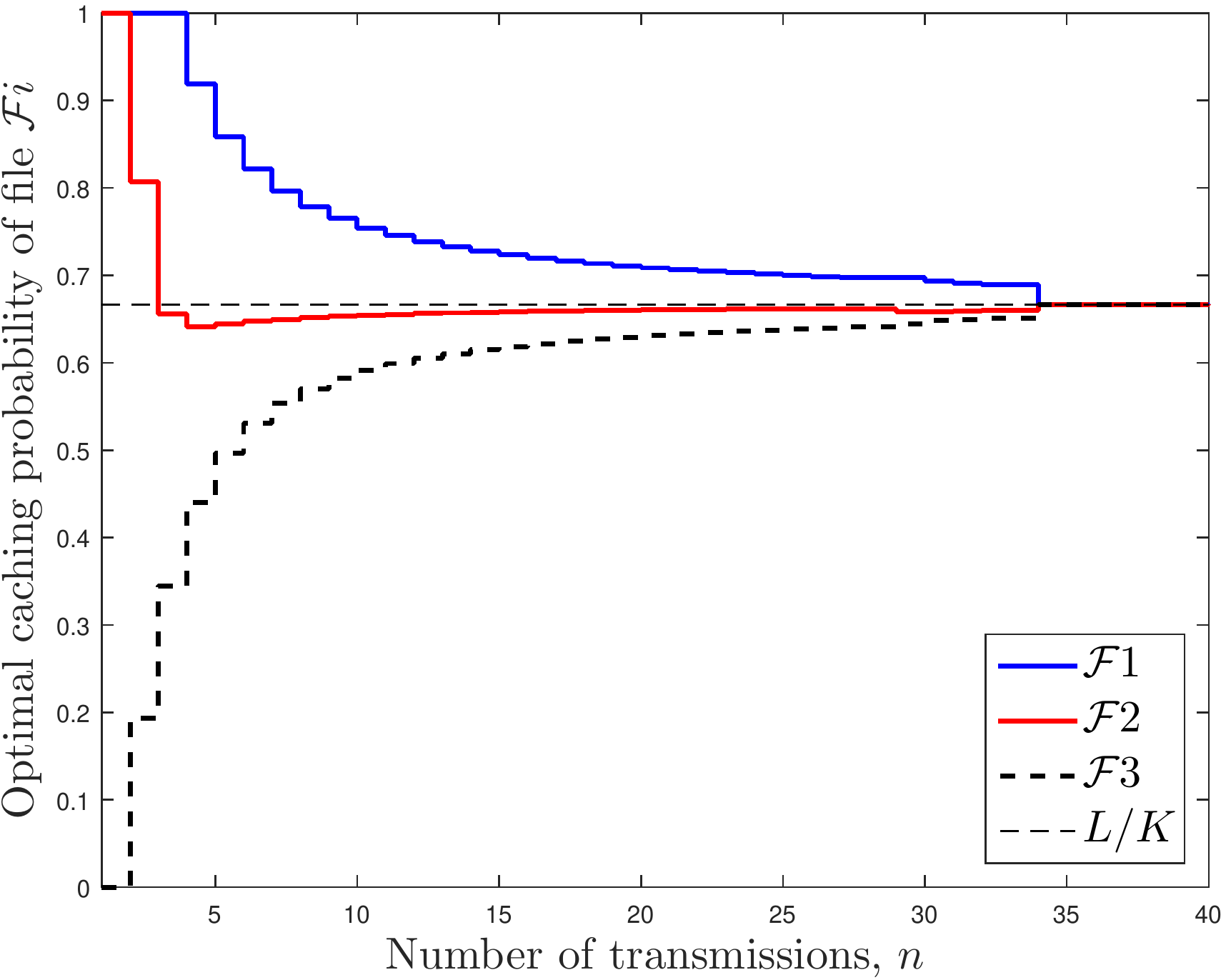}
\label{Fig:ThreeFileL2}
\end{subfigure}
\caption{ Effect of retransmissions on the optimal cache placement strategy for a library of $K = 3$ files with cache storage  $L =1$ {\em (left.)} and $L =2$ {\em (right).} }
\label{Fig:ThreeFile}
\end{figure}

\section{Conclusions}
In this paper, we have studied the effect of retransmissions on the the network performance (in terms of hit probability) for both static and mobile user scenarios in a cache-enabled SCBS network. Intuitively, if the user is allowed to access the caches multiple times (due to retransmissions), it becomes more likely that it will obtain the requested file, either due to better channel conditions or a new neighbourhood of cached SCBSs (due to user mobility). The analytical expressions developed also agree with the intuition and show an increase in hit probability with retransmissions. For a maximum $n$ transmission attempts, we also determine the optimal cache probabilities for both static and mobile user scenarios. The optimal solutions demonstrate that while it is optimal to cache the most popular content for a static user, SCBSs can cache content in a more balanced way for mobile users.
\label{Sec: sec5}

\begin{figure}[!t]
\centering
  \includegraphics[width=0.72 \linewidth]{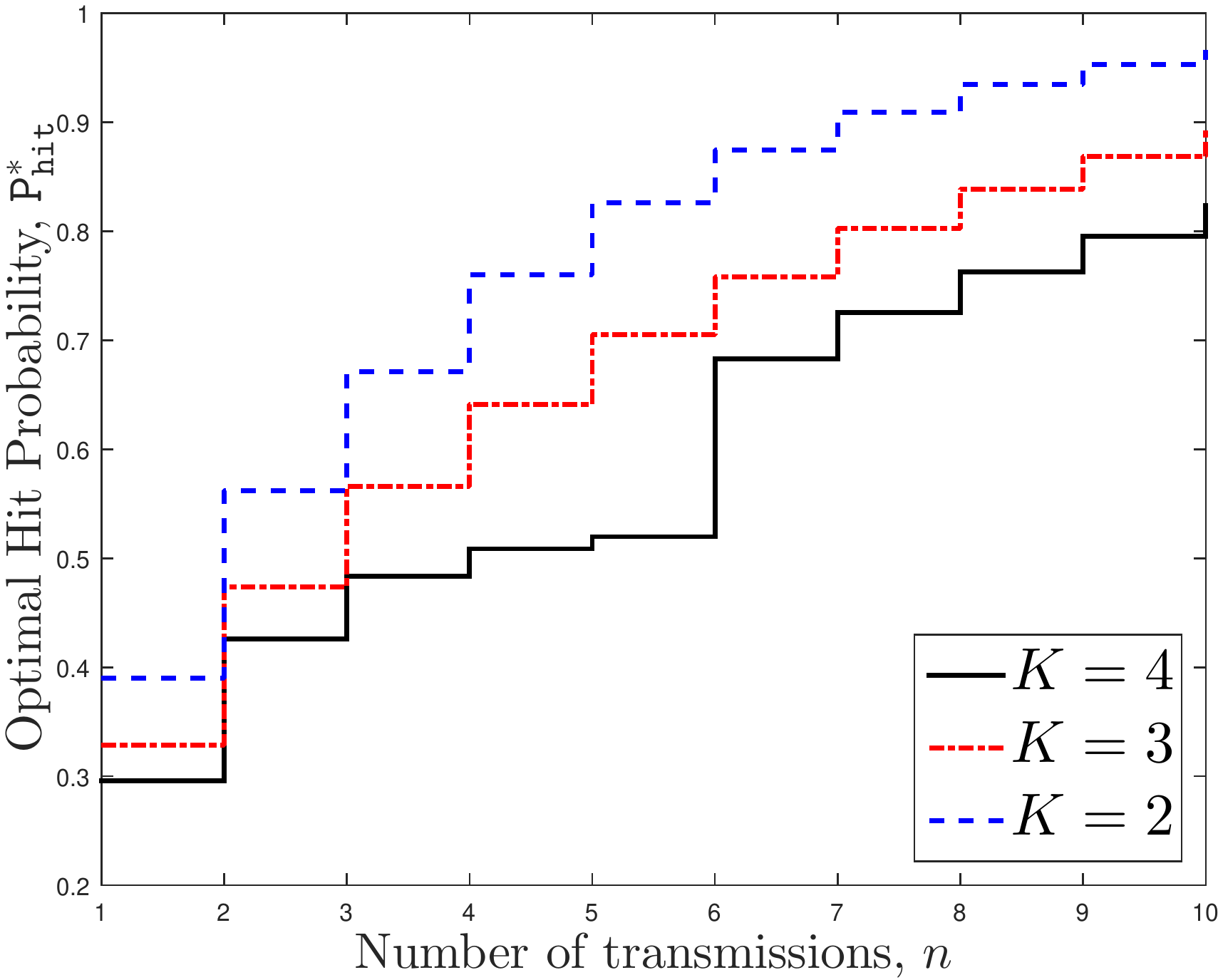}
  \caption{Effect of library size $K$ on the hit probability for a mobile user following policy 1 ($L = 1$).}
  \label{Fig: EffectK}
\end{figure}

\appendix
\subsection{Proof of Corollary \ref{Cor1}}
Using Equation (\ref{Eq:Golden1}) for the case of $K=2$ and $L = 1$, we obtain,
\begin{align}
&\bigg(\frac{-v^*}{n{{p}_{{\rm c}}^{(1)}}}\bigg)^{\frac{1}{n-1}} = \frac{2-{p}_{{\rm c}}^{(1)}}{\big(\frac{1}{{\tt P}_{{\rm R}_1}}\big)^{\frac{1}{n-1}} + \big(\frac{1}{{\tt P}_{{\rm R}_2}}\big)^{\frac{1}{n-1}}} 
\end{align} \label{Eq:2file}
Rearranging a few terms in the intervals of Theorem \ref{Theorem 1}, we obtain
\begin{align*}
b_1^{*} &= \left\{
     \begin{array}{lr}
       0, \qquad \qquad \qquad \qquad  (\frac{-\nu^*}{n{p}_{{\rm c}}^{(1)}})^{\frac{1}{n-1}} > {\tt P}_{{\rm R}_1}^{\frac{1}{n-1}} \\
       1, \qquad  \qquad \qquad \qquad  (\frac{-\nu^*}{n{p}_{{\rm c}}^{(1)}})^{\frac{1}{n-1}} < {\tt P}_{{\rm R}_1}^{\frac{1}{n-1}} (1-{p}_{{\rm c}}^{(1)})\\
      \frac{1}{{p}_{{\rm c}}^{(1)}}\big[ 1 - (\frac{1}{{\tt P}_{{\rm R}_1}})^{\frac{1}{n-1}}\big(\frac{-v^*}{n{p}_{{\rm c}}^{(1)}}\big)^{\frac{1}{n-1}}\big] , \qquad   \text{otherwise}
      \end{array}
   \right. \\\notag
&\stackrel{(a)}= \left\{
     \begin{array}{lr}
       0, \qquad \qquad \qquad \qquad \qquad 1-{p}_{{\rm c}}^{(1)} > (\frac{{\tt P}_{{\rm R}_1}}{{\tt P}_{{\rm R}_2}})^{\frac{1}{n-1}}\\
       1, \qquad  \qquad \qquad \qquad \qquad \frac{1}{1-{p}_{{\rm c}}^{(1)}} < (\frac{{\tt P}_{{\rm R}_1}}{{\tt P}_{{\rm R}_2}})^{\frac{1}{n-1}}\\
      \frac{1}{{p}_{{\rm c}}^{(1)}}\big[1 - \frac{2-{p}_{{\rm c}}^{(1)}}{1 + (\frac{{\tt P}_{{\rm R}_1}}{{\tt P}_{{\rm R}_2}})^{\frac{1}{n-1}}}\big] , \qquad   \text{otherwise}
      \end{array}
   \right. \\\notag
   &\stackrel{(b)}= \left\{
     \begin{array}{lr}
       0, \qquad \qquad \qquad \qquad \qquad n < 1+ \frac{\gamma}{log_2{(1-{p}_{{\rm c}}^{(1)})}}\\
       1, \qquad  \qquad \qquad \qquad \qquad n < 1 +\frac{\gamma}{\log_2({\frac{1}{1-{p}_{{\rm c}}^{(1)}}})}\\
      \frac{a-1+{p}_{{\rm c}_{1}}^{(1)}}{(a+1){p}_{{\rm c}_{1}}^{(1)}} , \qquad  \qquad \qquad \qquad \text{otherwise}
      \end{array}
   \right.
\end{align*}
where (a) follows by using Eq. (\ref{Eq:2file}) and rearranging a few terms. Step (b) is obtained by using the Zipf's law ${\tt P}_{{\rm R}_i} = {i^{-\gamma}}/{\sum_{j=1}^K j^{-\gamma}}$, where $\gamma > 0$ is the Zipf parameter and using $a = 2^{\frac{\gamma}{n-1}}$. The final result follows by ignoring the interval corresponding to $b_1^{*} = 0$ as it happens only when the number of transmissions $n<1$, which is not possible.

\subsection{Proof of Equation (\ref{Eq:JointSuccProbP1})}
For a static user scenario under policy 1, the typical user connects to the closest SCBS located at distance $R_1$ during all transmissions. From definition, the joint coverage probability of file ${\cal F}_i$ in $k$ transmissions is given as
\begin{align*}
{{\tt P}^{(1)}_{i,k}} &= \E_{R_1}\bigg[\P\bigg(\bigcap_{j \in \{1 \dots k\}} \frac{h_{xj} r_1^{-\alpha}}{\sum\limits_{y \in \Phi \setminus \{x\} } h_{yj} \|y\|^{-\alpha}} > T \big|r_1 \bigg )\bigg]\\
& \stackrel{(a)}= \E_{R_1}\bigg[\prod_{j=1}^{k}\exp\bigg(-Tr_1^{\alpha} \sum\limits_{y \in \Phi \setminus \{x\}} h_{yj} \|y\|^{-\alpha}\bigg)\bigg]\\
&\stackrel{(b)}= \E_{R_1}\bigg[\prod_{y \in \Phi \setminus \{x\}} \bigg(\frac{1}{1+Tr_1^{\alpha}\|y\|^{-\alpha}}\bigg)^k \bigg]
\end{align*}
where (a) results from the fact that $h_{xj} \sim \exp(1)$ and the i.i.d. fading assumption across the $k$ attempts resulting in simply the product of each term, (b) follows from $h_{yj} \sim \exp(1)$. The final result follows from the PGFL of PPP $\Phi$, converting from Cartesian to polar coordinates and deconditioning w.r.t. $R_1$.

\bibliographystyle{IEEEtran}
\bibliography{ref,ref-HD}
\end{document}